\documentclass[12pt, twoside]{article}

\usepackage{amsmath,amsthm,amssymb}

\usepackage{times}

\usepackage{enumerate}

\makeatletter

\def\author#1{\gdef\autrun{\def\and{\unskip, }#1}\gdef\@author{#1}}

\makeatother

\def\subjclass#1{{\renewcommand{\thefootnote}{}%

\footnote{\emph{Mathematics Subject Classification (2010):} #1}}}

 \newtheorem{theorem}{\bf Theorem}
 \newtheorem{lemma}[theorem]{\bf Lemma}
 \newtheorem{cor}[theorem]{\bf Corollary}

\newtheorem{remark}[theorem]{\bf Remark}
\newtheorem{definition}[theorem]{\bf Definition}

\def\blue\color{blue}

\usepackage{color}
\def\blue{\color{blue}}

\usepackage{url}

\begin{document}

\baselineskip=17pt

\author{Simeon Ball, Michel Lavrauw and Tabriz Popatia}

\title{Griesmer type bounds for additive codes over finite fields, integral and fractional MDS codes \footnote{4 December 2025. The first and third author  are supported by the Spanish Ministry of Science, Innovation and Universities grant PID2020-113082GB-I00 funded by MICIU/AEI/10.13039/501100011033 and grant PID2023-147202NB-I00.}}

\date{}

\maketitle


\begin{abstract}
In this article we prove Griesmer type bounds for additive codes over finite fields. These new bounds give upper bounds on the length of maximum distance separable (MDS) codes, codes which attain the Singleton bound. We will also consider codes to be MDS if they attain the fractional Singleton bound, due to Huffman. We prove that this bound in the fractional case can be obtained by codes whose length surpasses the length of the longest known codes in the integral case. For small parameters, we provide exhaustive computational results for additive MDS codes, by classifying the corresponding (fractional) subspace-arcs. This includes a complete classification of fractional additive MDS codes of size 243 over the field of order 9.\end{abstract}
\section{Introduction}

Let ${\mathbb F}_{q}$ denote the finite field with $q$ elements, where $q$ is a prime power. An {\em additive code} of length $n$ over ${\mathbb F}_{q^h}$ is a subset $C$ of ${\mathbb F}_{q^h}^n$ with the property that for all $u,v \in C$ the sum $u+v \in C$. It is easy to prove that an additive code is linear over some subfield, which we will assume to be ${\mathbb F}_q$. 
We use the notation 
$$[n,r/h,d]_q^h$$ 
to denote an additive code of length $n$ over ${\mathbb F}_{q^h}$, of size $q^r$ and minimum distance $d$, which is linear over ${\mathbb F}_q$. Our motivation comes from the fact that additive codes which are contained in their symplectic dual can be used to construct quantum stabiliser codes \cite{Rains1999},  \cite{KKKS2006}. Additive codes over small fields which have additional properties such as cyclicity have been studied in
\cite{Bierbrauer2007}, \cite{BE2000},
\cite{Danielsen2009}, \cite{Danielsen2012}, \cite{DP2006}, \cite{DR2005}. Most relevant to this article, where MDS codes are featured, are the articles  \cite{BGL2023}, \cite{Huffman2013} and \cite{YS2024}.
Recall that an additive code is {\em linear over ${\mathbb F}_{q^h}$} if $u \in C$ implies $\lambda u \in C$ for all $u \in C$ and all $\lambda \in {\mathbb F}_{q^h}$. In our notation here, a $[n,k,d]_q^1$ code is linear.
For an additive code $C$, the minimum distance $d$ is equal to the minimum weight, so this is equivalent to saying that each non-zero vector in $C$ has at least $d$ non-zero coordinates. 

The Griesmer bound \cite{Griesmer1960} for linear codes states that, if there is a $[n,k,d]_q^1$ code then
\begin{equation} \label{gb}
n \geqslant \sum_{j=0}^{k-1} \left \lceil  \frac{d}{q^j}\right  \rceil.
\end{equation}
This bound can be reformulated as follows,
\begin{equation} \label{gb2}
n \geqslant k+d-m+\sum_{j=1}^{m-1} \left \lceil  \frac{d}{q^j} \right \rceil,
\end{equation}
where $m\leqslant k$ is such that
$$
q^{m-2} < d \leqslant q^{m-1},
$$
or $m=k$ if $q^{k} < d$.
We will prove that similar bounds hold for additive codes over finite fields, see Theorem~\ref{addbound} and Theorem~\ref{addbound2}. We will also show that simply replacing $k$ by $\lceil r/h \rceil$ or $\lfloor r/h \rfloor$ in (\ref{gb}) does not lead to a valid bound for additive codes by constructing additive codes that invalidate these natural generalisations.
We will be particularly interested in codes of length $n$ and minimum distance $d$ which meet the Singleton bound
$$
|C|  \leqslant |A|^{n-d+1},
$$
which holds for all codes over an alphabet $A$, and is not limited to linear or additive codes. Codes which attain this bound are called {\it maximum distance separable codes}, or simply MDS codes. MDS codes are an important class of codes, which are implemented in many applications where we can allow $|A|$ to be large. 
As observed in previous articles, \cite[Theorem 10]{Huffman2013}, the Singleton bound for an $[n,r/h,d]_q^h$-code
can be reformulated as
$$
\lceil r/h \rceil  \leqslant n-d+1,
$$
since $n$ and $d$ are always integers, and we call codes which attain this bound for $r/h \not\in {\mathbb N}$, {\em fractional MDS} codes and for $r/h \in {\mathbb N}$, {\em integral MDS} codes.

\section{Additive Reed-Solomon codes and bounds on linear MDS codes}

The most commonly implemented MDS codes are the Reed-Solomon codes
$$
\{(f(a_1),\ldots,f(a_{q}),c_{k-1}) \ | \ f(x)=\sum_{i=0}^{k-1} c_ix^i,\  c_i \in \mathbb F_q \},
$$
where $a_1,\ldots,a_q$ denote the elements of $\mathbb F_q$ and $n=q+1$.
It follows immediately that $d=n-(k-1)$ since a polynomial of degree $k-1$ has at most $k-1$ zeros.

Our first observation is that Reed-Solomon codes extend to additive, not necessarily linear, codes as follows.
\begin{lemma}
If $S_0,\ldots, S_{k-1}$ are additive subsets of ${\mathbb F}_{q^h}$ and
$$
C=\{(f(a_1),\ldots,f(a_{q^h}),c_{k-1}) \ | \ f(x)=\sum_{i=0}^{k-1} c_ix^i,\  c_i \in S_i \},
$$
then $C$ is a fractional MDS code over ${\mathbb F}_{q^h}$ if and only if
$$
q^{kh} > \prod_{i=0}^{k-1} |S_i| > q^{(k-1)h}.
$$
\end{lemma}
\begin{proof}
By construction $|C|=\prod_{i=0}^{k-1} |S_i|$. Since the sets $S_i$ are additive, it follows that $C$ is an additive code, and the condition $q^{kh}>|C|=q^r>q^{(k-1)h}$ implies $\lceil r/h\rceil=k$, and so $C$ is a fractional MDS code since $d=n-(k-1)$.
\end{proof}
There are constructions of additive MDS codes which are based on restricting the evaluation set \cite{YS2024}, rather than restricting the set of coefficients as we have done here. In this paper, we are interested in finding the longest MDS codes and in particular, we would like to construct additive codes which are longer than their linear counterparts.

The Griesmer bound gives two important bounds for linear MDS codes. These results are well-known, but we list these as theorems since we will obtain similar bounds for additive MDS codes.

\begin{theorem} \label{dboundlinear}
If $k\geqslant 2$ and there is a $[n,k,d]^1_{q}$ linear MDS code then
$$
d \leqslant  q, \quad \mbox{ and } \quad 
n \leqslant q+k-1.
$$
\end{theorem}

\begin{proof}
The Griesmer bound for a linear $[n,k,d]_q^1$ MDS code $C$ can be rewritten as
$$
n-d=k-1 \geqslant  \sum_{j=1}^{k-1} \left \lceil  \frac{d}{q^j} \right \rceil.
$$
It follows that if $k \geqslant 2$ then $d \leqslant q$ and so, since $C$ is MDS,
$n=d+k-1 \leqslant q+k-1$.
\end{proof}

\begin{theorem} \label{kboundlinear}
If $n\geqslant k+2$ and there is a $[n,k,d]_{q}^1$ linear MDS code then
$$
k \leqslant  q-1.
$$
\end{theorem}

\begin{proof}
Since the dual of a linear MDS code is a linear $[n,n-k,k+1]$ MDS code, Theorem~\ref{dboundlinear} implies that if $n \geqslant k+2$ then $k \leqslant q-1$.
\end{proof}

Interestingly, we will see that these restrictions do not carry over to additive codes. We will provide examples which better these bounds. Motivated by this fact we prove Griesmer type bounds for additive codes and then consider the consequences of these bounds for MDS codes. We will obtain similar bounds to Theorem~\ref{dboundlinear} and Theorem~\ref{kboundlinear}, which will apply to additive codes, in Theorem~\ref{dbound} and Theorem~\ref{kbound}.

\section{Additive codes over finite fields}

An $[n,k,d]_q$ linear code $C$ can be defined as the rowspace of a $k \times n$ matrix called a {\em generator matrix} for $C$. 
An $[n,r/h,d]_{q}^h$ additive code $C$ can then be defined as the ${\mathbb F}_q$-space spanned by the rows of ${G}$, an $r \times n$ matrix with entries in ${\mathbb F}_{q^h}$, again called a {\em generator matrix} for $C$. We will always set $k=\lceil r/h \rceil$.

The elements of ${G}$ are elements of ${\mathbb F}_{q^h}$, which we can write out over a basis $\mathcal B$
for ${\mathbb F}_{q^h}$ over ${\mathbb F}_q$ to obtain a $r \times nh$ matrix $\tilde G$ whose elements are from ${\mathbb F}_q$, and whose columns are grouped together in $n$ blocks of $h$ columns.
Each block of $h$ columns spans a projective subspace of PG$(r-1,q)$ of (projective) dimension at most $h-1$. We define ${\mathcal{X}}_G(C)$ to be the multiset of subspaces which we obtain from the $n$ blocks of $h$ columns of $\tilde G$ in this way. The fact that ${\mathcal{X}}_G(C)$ is independent of the choice of the basis $\mathcal B$ is easily seen by observing that each element of the multiset ${\mathcal{X}}_G(C)$ is obtained as the column space of one of the blocks in $\tilde G$, which is an $r\times h$ matrix. A change of basis might change the matrix, but not its column space. Unless the choice of the generator matrix is not clear from the context, we usually just write ${\mathcal X}(C)$ instead of ${\mathcal{X}}_G(C)$.

\begin{definition}
A projective $h-(n,r,d)_q$ system is a multiset $\mathcal S$ of $n$ subspaces of ${\mathrm{PG}}(r-1,q)$ of dimension at most $h-1$ such that each hyperplane of ${\mathrm{PG}}(r-1,q)$ contains at most $n-d$ elements of $\mathcal S$ and some hyperplane contains exactly $n-d$ elements of $\mathcal S$.
\end{definition}

\begin{theorem}\label{thm:proj_system}
If $C$ is an additive $[n,r/h,d]_{q}^h$ code, then ${\mathcal{X}}(C)$ is a projective $h-(n,r,d)_q$ system, and conversely, each projective $h-(n,r,d)_q$ system defines an additive $[n,r/h,d]_{q}^h$ code.
\end{theorem}
\begin{proof}
The proof is essentially the same as the proof for the integral case, see Theorem 4 in \cite{BGL2023}.
If $H$ is a hyperplane of ${\mathrm{PG}}(r-1,q)$ with dual coordinates $a\in {\mathbb{F}}_q^r$, then the fact that the codeword $aG$ has weight at least $d$ is equivalent to the condition that $H$ contains at most $n-d$ elements of ${\mathcal{X}}(C)$. 
\end{proof}

We say an additive code $C$ is {\em unfaithful} if not all the elements of $\mathcal X(C)$ are of dimension $h-1$, and {\em faithful} otherwise.  

\begin{remark} \label{onlyrem}
Note that one can always extend the subspaces of $\mathcal X(C)$ so that they have dimension $h-1$. By Theorem \ref{thm:proj_system}, this can be done arbitrarily without the minimum distance decreasing. Thus, an unfaithful additive $[n,r/h,d]_{q}^h$ code $C$ can always be converted in a faithful additive $[n,r/h, \geqslant d]_{q}^h$ code.
\end{remark}

The above notion of an unfaithful code is defined in terms of the corresponding projective system. The coding-theoretic interpretation of the notion is as follows.

\begin{lemma}
An additive code $C$ over ${\mathbb{F}}_{q^h}$ is unfaithful if and only if there is a coordinate position in which the values of ${\mathbb{F}}_{q^h}$, which appear in the codewords of $C$, are contained in an ${\mathbb{F}}_q$-subspace of dimension at most $h-1$.
\end{lemma}
\begin{proof}
Let $C$ be an additive code with generator matrix $G$, and corresponding projective system ${\mathcal{X}}(C)$. Let $\tilde G$ denote the $r\times nh$ matrix obtained from $G$ as described above.

Suppose that $C$ is unfaithful and that the $i$-th block $B$ of $h$ columns of $\tilde G$ determines an element of ${\mathcal{X}}(C)$ of dimension at most $h-2$. Then the $h$ columns of $B$ are linearly dependent, i.e. there exists a nonzero $y\in {\mathbb{F}}_q^h$ such that $By^T=0$. In other words, the rows of $B$, considered as vectors in ${\mathbb{F}}_q^h$ are orthogonal to some fixed nonzero vector $y$. Since the $i$-th coordinate $c_i\in {\mathbb{F}}_{q^h}$ of a codeword $c$ of $C$, is a linear combination of the rows of $B$, when considered as a vector in ${\mathbb{F}}_q^h$, the claim follows.
The reverse implication follows by reversing the above argument.
\end{proof}


\begin{remark}
Consider an additive subcode $C_0$ of an additive code $C$. We can choose a basis for the code $C$ so that the first $r_0$ rows of $\tilde{G}$ span $C_0$. The subspaces of $\mathcal X(C_0)$ are obtained from the subspaces of $\mathcal X(C)$ by projecting from $\pi$, the subspace spanned by the last $r-r_0$ coordinates. Then $C_0$ is faithful if and only if $C$ is faithful and all the elements of $\mathcal X(C)$ do not intersect, are skew to, $\pi$.
\end{remark}

\section{A bound for additive codes over finite fields} \label{boundsec}

The geometric approach described in the previous section will be used to prove the following bound for additive codes.

\begin{theorem} \label{addbound}
If there is a $[n,r/h,d]_{q}^h$ additive code  then
$$
n \geqslant \lceil r/h \rceil +d-m+ \left\lceil  \frac{d}{f(q,m)} \right\rceil,
$$
where $r=(k-1)h+r_0$, $k=\lceil r/h \rceil $, $1 \leqslant r_0 \leqslant h$,
$$
f(q,m)=\frac{q^{(m-2)h+r_0}(q^h-1)}{q^{(m-2)h+r_0}-1}
$$
for all $m$ such that $2 \leqslant m \leqslant k$ and is maximised when $m$ is such that
$$
q^{(m-2)h+r_0} < d \leqslant q^{(m-1)h+r_0}
$$
and $m=k$ if 
$$
d>q^{r}.
$$
\end{theorem}

\begin{proof}

 Let $C$ be an additive $[n,r/h,d]_{q}^h$ code, with generator matrix $G$, and let $\tilde G$ and $\mathcal X(C)$ be as described above.
 
By performing elementary row operations, we can suppose that $\tilde G$ is in row-echelon form. Let $m$ be an integer such that $2 \leqslant m \leqslant k$. Consider the sub-matrix of $\tilde G$ formed by the last $(m-1)h+r_0$ rows. Since $\tilde G$ is in row-echelon form the first 
$$
s =(r-(m-1)h-r_0)/h=k-m
$$
blocks of $h$ columns of this sub-matrix are all zero.

Let $\tilde{G}_1$ be the submatrix of $\tilde{G}$ formed by the last $(m-1)h+r_0$ rows and the last $n-s$ blocks of columns. Then $\tilde{G}_1$ generates an additive $[n-s,m-1+r_0/h,\geqslant d]_q^h$ code $C_1$. 
 
Let $\mathcal X_1$ be the multi-set of size $n-s$ of the subspaces of dimension at most $h-1$ in PG$((m-1)h+r_0-1,q)$ spanned by the $n-s$ blocks of columns of $\tilde{G}_1$.
 
A hyperplane of PG$((m-1)h+r_0-1,q)$ contains at most $n-s-d$ of the subspaces of $\mathcal X_1$, since the non-zero codewords $u\tilde{G}_1$ of $C_1$ have weight at least $d$.

Each subspace of $\mathcal X_1$ is contained in at least 
$$
(q^{(m-2)h+r_0}-1)/(q-1)
$$
hyperplanes. Since there are 
$$
(q^{(m-1)h+r_0}-1)/(q-1)
$$
hyperplanes in PG$((m-1)h+r_0-1,q)$ 
we have
\begin{equation} \label{eq1}
(q^{(m-2)h+r_0}-1)(n-s) \leqslant (q^{(m-1)h+r_0}-1)(n-s-d).
\end{equation}
This implies
$$
n \geqslant s+d+\left \lceil  \frac{d}{f(q,m)}\right  \rceil.
$$
from which the bound follows using $s=k-m$.
 
 It remains to prove that the bound is maximised by choosing $m=m_0$, where $m_0$ is such that 
 $$
q^{(m_0-2)h+r_0} < d \leqslant q^{(m_0-1)h+r_0}. 
$$

Let
$$
B(m)=\lceil r/h \rceil +d-m+ \left \lceil  \frac{d}{f(q,m)} \right \rceil.
$$
Then
$$
B(m+1)-B(m)=-1+\left\lceil  \frac{d}{f(q,(m+1))}\right \rceil-\left\lceil  \frac{d}{f(q,m)} \right\rceil
$$
and therefore
$$
-1+\left \lceil  \frac{d}{f(q,(m+1))} -\frac{d}{f(q,m)}\right \rceil-1 \leqslant 
B(m+1)-B(m) \leqslant -1+\left\lceil  \frac{d}{f(q,(m+1))} -\frac{d}{f(q,m)}\right\rceil.
$$
For $m \geqslant m_0$ we have that
$$
B(m+1)-B(m)\leqslant -1+\left\lceil  \frac{d}{q^{(m-1)h+r_0}} \right\rceil = -1+1=0.
$$
For $m \leqslant m_0-1$ we have that
$$
B(m+1)-B(m)
\geqslant 
-2+\left \lceil  \frac{d}{q^{(m-1)h+r_0}} \right \rceil \geqslant -2+2=0.
$$
This shows that the bound is maximised for $m=m_0$.
\end{proof}

In the integral case the bound in Theorem~\ref{addbound} gives a bound similar to (\ref{gb2}), the Griesmer bound for linear codes. Observe that the only difference from the Griesmer bound is that the ceiling function appears outside the sum and not inside.

\begin{cor}
If there is a $[n,k,d]_{q}^h$ additive code  then
$$
n \geqslant k+d-m+ \left\lceil  \sum_{j=1}^{m-1} \frac{d}{q^{jh}} \right\rceil,
$$
where $k \in {\mathbb N}$ and $m$ is such that
$$
q^{(m-2)h} < d \leqslant q^{(m-1)h} \leqslant q^r. 
$$
\end{cor}

\begin{proof}
Since
$$
\frac{1}{f(q,m)}=\frac{1}{q^{(m-1)h}}(1+q^h+\cdots+q^{(m-2)h})= \sum_{j=1}^{m-1} \frac{1}{q^{jh}},
$$
the bound follows from Theorem~\ref{addbound}.
\end{proof}

In the non-integral (fractional) case that $r/h \not\in {\mathbb N}$, we get a slightly weaker bound using the same reasoning as above.

\begin{cor} \label{addboundcor}
If there is a $[n,r/h,d]_{q}^h$ additive code then
$$
n \geqslant \lceil r/h \rceil +d-m+ \left\lceil  \sum_{j=1}^{m-2} \frac{d}{q^{jh}} \right\rceil,
$$
where $r=(k-1)h+r_0$, $k=\lceil r/h \rceil $, $1 \leqslant r_0 \leqslant h$, and $m$ is such that
$$
q^{(m-2)h+r_0} < d \leqslant q^{(m-1)h+r_0} \leqslant q^r.
$$
\end{cor}

\begin{proof}
Since
$$
\frac{1}{f(q,m)}>\frac{1}{q^{(m-2)h+r_0}} \frac{q^{(m-2)h+r_0}-q^{r_0}}{q^h-1}=\frac{1}{q^{(m-2)h}}(1+q^h+\cdots+q^{(m-3)h})= \sum_{j=1}^{m-2} \frac{1}{q^{jh}},
$$
the bound follows from Theorem~\ref{addbound}.
\end{proof}

\section{A second bound for additive codes over finite fields}

Inspired by the bound for $q=4$ from \cite{BMP2021} and \cite{GLLM2023}, in this section we prove another bound for additive codes over finite fields. We then show that this new bound is generally not as good as the bound in Theorem~\ref{addbound} for $d<q^r$.

\begin{theorem} \label{addbound2}
If there is a $[n,r/h,d]_q^h$ additive code then
$$
n \geqslant d+\frac{q-1}{q^h-1}\sum_{j=1}^{r-h} \lceil \frac{d}{q^j} \rceil.
$$
\end{theorem}

\begin{proof}
Let $\{e_1,\ldots,e_h\}$ be a basis for ${\mathbb F}_{q^h}$ over ${\mathbb F}_q$ and let $G$ be an $r \times n$ matrix over ${\mathbb F}_{q^h}$ whose ${\mathbb F}_q$ row span is an additive $[n,r/h,d]_q^h$ code. Then we can write
$$
G=\sum_{i=0}^h G_ie_i
$$
where $G_i$ is a $r \times n$ matrix over ${\mathbb F}_q$. 


Let $\hat{G}_j$, $j \in \{1,\ldots,(q^h-1)/(q-1)\}$, be the distinct elements of the projective subspace
$$
\langle G_1,\ldots,G_h\rangle.
$$
In other words, the matrix
$$
\hat{G}_j=\sum_{i=1}^h c_{ij}G_i
$$
for some point $c_j=(c_{1j},\ldots,c_{hj})$ in $PG(h-1,q)$. 

Consider the linear code generated by the $r \times (n(q^h-1)/(q-1))$ matrix
$$
\hat{G}=(\hat{G}_1,\ldots,\hat{G}_{(q^h-1)/(q-1)}).
$$

Let $u \in {\mathbb F}_q^r$. 

If the $m$-th coordinate of the codeword $uG$ is zero then the $m$-th coordinate of $uG_i$ is zero, for all $i \in \{1,\ldots,h\}$, and therefore the $m$-th coordinate of $u\hat{G}_j$ is zero for all $j \in \{1,\ldots,(q^h-1)/(q-1)\}$. 

Suppose that the $m$-th coordinate of the codeword $uG$ is non-zero. Then 
$$
a=(uG_1,\ldots,uG_h)
$$
is a non-zero vector of ${\mathbb F}_q^h$. If the point $c_j$ is orthogonal to $a$ then the $m$-th coordinate of $u\hat{G}_j$ is zero. Since $c_j$ runs through all points in $PG(h-1,q)$, the $m$-th coordinate of $u\hat{G}_j$ is zero for $(q^{h-1}-1)/(q-1)$ values of $j \in \{1,\ldots,(q^h-1)/(q-1)\}$ and therefore non-zero for $q^{h-1}$ coordinates.

Thus, if the weight of $uG$ is $w$ then weight of $u\hat{G}$ is $wq^{h-1}$. 

Therefore, $\hat{G}$ generates a $[\frac{(q^h-1)}{q-1}n,r,dq^{h-1}]$ linear code over ${\mathbb F}_q$.

Applying the Griesmer bound for linear codes gives
$$
\frac{(q^h-1)}{q-1}n \geqslant \sum_{j=0}^{r-1} \lceil \frac{dq^{h-1}}{q^j} \rceil=\frac{(q^h-1)}{q-1}d+\sum_{j=1}^{r-h} \lceil \frac{d}{q^j} \rceil .
$$
\end{proof}

\section{Comparing the bounds in Theorem~\ref{addbound} and Theorem~\ref{addbound2}}

In this section we will compare the bounds in Theorem~\ref{addbound} and Theorem~\ref{addbound2}.

\begin{theorem} \label{comparison}
If 
$d$ is such that
$$
q^{(m-2)h+r_0}  < d \leqslant q^{(m-1)h+r_0},
$$
for some $m \in \{2,\ldots,\lceil r/h\rceil\}$ and
$$
(\lceil r/h\rceil -m)(q^h-1) \geqslant (r-h)(q-1)+q^{h}-q
$$
then the bound in Theorem~\ref{addbound} is better than the bound in Theorem~\ref{addbound2}.
\end{theorem}

\begin{proof}
Let 
$$
b_1= \lceil r/h \rceil +d-m+ \left\lceil  \frac{d}{f(q,m)} \right\rceil,
$$
the bound in Theorem~\ref{addbound} and
let
$$
b_2= d+\frac{q-1}{q^h-1}\sum_{j=1}^{r-h} \lceil \frac{d}{q^j} \rceil,
$$
the bound in Theorem~\ref{addbound2}.

Since $d \leqslant q^{(m-1)h+r_0}$, and recalling that $r=(k-1)h+r_0$ and $k=\lceil r/h\rceil$,
$$
b_2 \leqslant d + ((k-m-1)h+1)\frac{q-1}{q^h-1}+\frac{q-1}{q^h-1}\sum_{j=1}^{(m-1)h+r_0-1} \lceil \frac{d}{q^j} \rceil.
$$
and so
$$
b_2 \leqslant d + (r-h)\frac{q-1}{q^h-1}+\frac{q-1}{q^h-1}\sum_{j=1}^{(m-1)h+r_0-1} \frac{d}{q^j}
$$
and
$$
b_2 \leqslant d + (r-h)\frac{q-1}{q^h-1}+\frac{d(q^{h-1}-1)}{q^{(m-1)h+r_0-1}(q^h-1)}+\frac{q-1}{q^h-1}\sum_{j=1}^{(m-2)h+r_0} \frac{d}{q^j}
$$
Meanwhile,
$$
b_1 \geqslant k+d-m+d\frac{(q^{(m-2)h+r_0}-1)}{(q^h-1)q^{(m-2)h+r_0}},
$$
and so
$$
b_1 \geqslant k+d-m+\frac{q-1}{q^h-1}\sum_{j=1}^{(m-2)h+r_0} \frac{d}{q^j}.
$$

If $b_1<b_2$ then
$$
k-m<(r-h)\frac{q-1}{q^h-1}+\frac{d(q^{h-1}-1)}{q^{(m-1)h+r_0-1}(q^h-1)}<
(r-h)\frac{q-1}{q^h-1}+\frac{q(q^{h-1}-1)}{(q^h-1)}.
$$

\end{proof}


\section{Bounds on additive MDS codes over finite fields}

Theorem \ref{addbound} has the following corollary for MDS codes.

\begin{theorem} \label{dbound}
If there is a $[n,r/h,d]_{q}^h$ additive MDS code then
$$
d \leqslant  q^{h}-1+\frac{q^h-1}{q^{r_0}-1}
$$
and
$$
n \leqslant  k -2+ q^{h}+\frac{q^h-1}{q^{r_0}-1},
$$
where $k=\lceil r/h\rceil$, $r=(k-1)h+r_0 >h$ and $1 \leqslant r_0 \leqslant h$.

\end{theorem}

\begin{proof}
An additive MDS code implies that $n=\lceil r/h \rceil +d-1$. Therefore, Theorem~\ref{addbound} implies
$$
\left \lceil  \frac{d}{f(q,m)} \right \rceil \leqslant m-1
$$
where
$$
f(q,m)=\frac{q^{(m-2)h+r_0}(q^h-1)}{q^{(m-2)h+r_0}-1}
$$
and for all $m \in \{2,\ldots,k\}$.

Since $r>h$ implies $k \geqslant 2$, we can set $m=2$ and the first inequality implies
$$
d \leqslant f(q,2).
$$
The bound on $n$ then follows by substituting in the (fractional) Singleton bound.
\end{proof}

We can improve on these bounds in the case that the MDS code is unfaithful.

\begin{theorem} \label{degbound}
If there is an unfaithful $[n,r/h,d]_{q}^h$ additive MDS code then
$$
d \leqslant  q^{h}+\frac{q^h-q^{r_0+1}}{q^{r_0}-1}
$$
and
$$
n \leqslant q^{h}+ k -1 +\frac{q^h-q^{r_0+1}}{q^{r_0}-1},
$$
where $r=(k-1)h+r_0 >h$ and $1 \leqslant r_0 \leqslant h$.
\end{theorem}

\begin{proof}
In the proof of Theorem \ref{addbound}, we have that $m=2$ and $s=k-2$, and at least one of the subspaces in $\mathcal X_1$ has dimension at most $h-2$. Thus, the inequality (\ref{eq1}) becomes
$$
(q^{r_0}-1)(n-(k-2)-1)+(q^{r_0+1}-1) \leqslant (q^{h+r_0}-1)(n-(k-2)-d).
$$
Since $d=n-k+1$, we have
$$
d \leqslant \frac{q^{h+r_0}-q^{r_0+1}}{q^{r_0}-1}=q^{h}+\frac{q^h-q^{r_0+1}}{q^{r_0}-1}.
$$
\end{proof}

Recall that in Theorem~\ref{dboundlinear} we had that $d \leqslant q^h$, so the bound in Theorem \ref{dbound} is the same when $r_0=h$ and slightly weaker when $r_0 \neq h$. 
In Theorem~\ref{kboundlinear}, we saw that the Griesmer bound allowed us to bound $k$ for linear MDS codes, as long we assumed that $n \geqslant k+2$, or equivalently, assuming $d \geqslant 2$. 
For additive MDS codes we obtain a similar bound, with roughly the same hypothesis.

\begin{theorem}    \label{kbound}
If there is an  $[n,k,d]_{q}^h$ additive MDS code $C$ for some $k \in {\mathbb N}$ and $n \geqslant k+2$ then
$$
k \leqslant  q^h-1.
$$
\end{theorem}

\begin{proof}

Since $n \geqslant k+2$, we can truncate the code $C$ to an MDS code $C_0$ of length $k+2$. 

Since $d=3$, any two codewords must be distinct on one of the first $k$ coordinates, thus there are sets $S_i={\mathbb F}_{q^h}$ and maps 
$f$ and $g$ from
$$
(S_1,\ldots,S_{k}) \rightarrow {\mathbb F}_{q^h}
$$
such that
$$
C_0=\{ (u_1,u_2,\ldots,u_{k},f(u_1,\ldots,u_{k}),g(u_1,\ldots,u_{k})) \ | \ u_i \in S_i \}.
$$

Since $C_0$ is additive ${\mathbb F}_q$-linear, the maps $f$ and $g$ must be additive ${\mathbb F}_q$-linear maps.

For each $i \in \{1,\ldots,k\}$, let
$$
\pi_i=\{(f(0,\ldots,0,u_i,0,\ldots,0),g(0,\ldots,0,u_i,0,\ldots,0)) \ | \ u_i \in S_i \}.
$$
Since $f$ and $g$ are additive ${\mathbb F}_q$ linear maps, $\pi_i$ is an ${\mathbb F}_q$ subspace of ${\mathbb F}_{q^h}^2$. 

For $i \neq j$, the distance between 
$$
u=(0,\ldots,0,u_i,0,\ldots,0,f(0,\ldots,0,u_i,0,\ldots,0),g(0,\ldots,0,u_i,0,\ldots,0))
$$ 
$$
v=(0,\ldots,0,v_j,0,\ldots,0,f(0,\ldots,0,v_j,0,\ldots,0),g(0,\ldots,0,v_j,0,\ldots,0))
$$
is at least $3$, which implies
$$
\pi_i \cap \pi_j = \{ 0\}.
$$
Since, $u$ has weight at least $3$, $\pi_i$ has trivial intersection with both
$$
\pi_0=\{(0,x) \ | \ x \in {\mathbb F}_{q^h} \}
$$
and
$$
\pi_{\infty}=\{(x,0) \ | \ x \in {\mathbb F}_{q^h} \}.
$$
for all $i \in \{1,\ldots,k\}$.

Therefore, we have that 
$$
\sum_{i=1}^{k} (|S_i|-1) \leqslant q^{2h}-2q^h+1.
$$
Since $|S_i|=q^h$, we have that
$$
k(q^h-1) \leqslant  q^{2h}-2q^h+1.
$$
\end{proof}

\section{The dual of an additive code}

The dual of an additive $[n,r/h,d]_{q}^h$ code $C$ is defined by
$$
C^{\perp}=\{ v \in {\mathbb F}_{q^h}^n \ | \ \mathrm{tr}_{q^h \mapsto q} (u \cdot v)=0, \ \mathrm{for} \ \mathrm{all} \ u \in C\},
$$
where $ \mathrm{tr}_{q^h \mapsto q}$ denotes the trace function from ${\mathbb F}_{q^h}$ to ${\mathbb F}_q$.
Since an ${\mathbb F}_q$-basis for $C$ defines $r$ equations for the ${\mathbb F}_q$-subspace $C^{\perp}$, the dual code is an additive $[n,n-r/h,d^{\perp}]_{q}^h$ code. Note that as an ${\mathbb F}_q$-vector space the vector space ${\mathbb F}_{q^h}^n$ has dimension $hn$ and so $|C^{\perp}|=q^{nh-r}$.

\begin{theorem} \label{degthm}
Let $C$ be an additive $[n,r/h,d]_{q}^h$ code, and let $d^\perp$ denote the minimal distance of $C^\perp$. Then $C$ is unfaithful if and only if $d^\perp=1$.
\end{theorem}

\begin{proof}
Recall that in Section~\ref{boundsec}, we defined a multiset of subspaces $\mathcal X(C)$ of PG$(r-1,q)$ from an additive code $C$ of size $q^r$, corresponding to the blocks of size $h$ of the expanded matrix $\tilde G\in {\mathbb F}_q^{r\times nh}$ obtained from the generator matrix $G\in {\mathbb F}_{q^h}^{r\times n}$ of $C$.
Without loss of generality we may assume that the first block of $h$ columns of $\tilde G$ spans a subspace of dimension at most $h-2$. Then the rank of this $r \times h$ matrix is strictly smaller than $h$ and so its columns are linearly dependent. This implies that there exists a nonzero element $x \in {\mathbb F}_{q^h}$ such that $(x,0,\ldots,0) \in C^{\perp}$.

To prove the other direction, assume $d^\perp=1$. Without loss of generality assume that 
$(x,0,\ldots,0) \in C^{\perp}$,
with $x \in {\mathbb F}_{q^h}\setminus\{0\}$.
This implies that the first block of $\tilde G$ consist of $h$ columns which are linearly dependent, and therefore the corresponding element of $\mathcal X(C)$ has dimension at most $h-2$. Hence
$C$ is unfaithful.
\end{proof}

\begin{theorem}
Let $C$ be an additive $[n,r/h,d]_{q}^h$ code, and let $d^\perp$ denote the minimal distance of $C^\perp$. 
Then $C$ is faithful with $d \geqslant 2$ if and only if $C^{\perp}$ is faithful with $d^{\perp} \geqslant 2$.
\end{theorem}

\begin{proof}
If $C$ is unfaithful then by Theorem~\ref{degthm}, $d^{\perp}=1$. Hence $C$ is faithful. The other direction follows by replacing $C$ with $C^{\perp}$.
\end{proof}
The following is \cite[Theorem 4.3]{BGL2023}, see also \cite[Theorem 3.3]{YS2024}.

\begin{theorem} \label{integraldual}
The dual of an integral additive MDS code is an additive MDS code.
\end{theorem}

In the fractional case, as pointed out in \cite[Example 3.1]{YS2024} the dual of a fractional MDS code is not necessarily MDS. Here, we aim to give a precise condition on when the dual of a fractional MDS code is MDS.

\begin{definition}
Let $J$ be a subset of $\{1,\ldots,n\}$. The {\em geometric quotient at $J$} of an additive $[n,r/h,d]_{q}^h$ code $C$  is a code $C/J$ of length $n-|J|$ defined as
$$
C/J=\{ u_J \ | \ u \in C \mbox{, and } u_j=0, \ \mathrm{for}\ \mathrm{all} \ j \in J\} \subseteq {\mathbb{F}}_{q^h}^{n-|J|},
$$
where $u_J$ denotes the vector obtained from $u$ by deleting the coordinate positions belonging to the subset $J$. We say a geometric quotient is {\em non-obliterating} if $|C/J|\geqslant q^{h}$. Note that an obliterating quotient always yields an unfaithful additive code.

\end{definition}
The code $C/J$ is sometimes referred to as a projection or a shortening of the code $C$, which can be confused with puncturing or truncation. The motivation for calling it the geometric quotient is the following theorem. This theorem generalises \cite[Lemma 3.11]{AB2023} from the integral case that $r/h \in {\mathbb N}$ to the fractional case. 

\begin{theorem}
Let $C$ be an additive $[n,r/h,d]_{q}^h$ code with projective system $\mathcal X(C)$. If $J$ is a subset of $\{1,\ldots,n\}$ and $U_J$ is the subspace of ${\mathrm{PG}}(r-1,q)$ spanned by the elements of $\mathcal X(C)$ corresponding to the positions in $J$, then $C_J$ is an additive code of length $n_J=n-|J|$, with ${\mathbb{F}}_q$-dimension $r_J=r-\dim U_J$, and with projective system in 
 the quotient geometry 
 $${\mathrm{PG}}(r-1,q)/U_J\cong {\mathrm{PG}}(r-\dim U_J -1,q)$$ 
 consisting of the subspaces $\langle A,U_J\rangle/U_J$, $A\in {\mathcal{X}}(C)\setminus U_J$.
\end{theorem}
\begin{proof}
Let $C$ be an additive $[n,r/h,d]_{q}^h$ code with projective system $\mathcal X(C)$, and with generator matrix $G$ consisting of columns $v_1,\ldots, v_n\in {\mathbb{F}}_{q^h}^r$. 

First suppose $J=\{j\}$. Consider any nonzero codeword $aG$ of $C$, $a\in {\mathbb{F}}_{q}^r$. Then $aG\in C/J$ if and only if $v_j\in a^\perp$.
Let $\chi$ denote the extension of the map used to map the columns of $G$ to the elements of ${\mathcal{X}}(C)$, with respect to some fixed ${\mathbb{F}}_q$-basis for ${\mathbb{F}}_{q^h}$. Then $\chi(v_j)\in {\mathcal{X}}(C)$ is a subspace of ${\mathbb{F}}_{q}^r$ of dimension $m_j\leqslant h$, and since $a\in {\mathbb{F}}_{q}^r$, the condition $v_j\in a^\perp$ is equivalent to the condition that $a^\perp$ contains $\chi(v_j)$.
So the code $C/J$  is an additive code of dimension $r-m_j$ over ${\mathbb{F}}_q$, and its codewords are of the form $aG$ where $a^\perp$ contains $\chi(v_j)$, i.e. $a^\perp + \chi(v_j)$ is a hyperplane in the quotient space ${\mathbb{F}}_{q}^r/\chi(v_j)$. Consider the map
$$
\eta~:~{\mathbb{F}}_{q^h}^r\rightarrow {\mathbb{F}}_{q^h}^{r-m_j}~:~v \mapsto \chi^{-1}(\varphi(\chi(v)+\chi(v_j)))
$$
where $\varphi$ denotes the natural isomorphism ${\mathbb{F}}_{q}^r/\chi(v_j)\cong {\mathbb{F}}_q^{r-m_j}$, and let $G_J$ denote the $(r-m_j)\times (n-1)$ matrix with columns $\eta(v_i)$, $i\neq j$. Then $G_J$ is a generator matrix for $C_J$, and the multiset ${\mathcal{X}}(C_J)$ consists of the subspaces $\langle \chi(v_i),\chi(v_j)\rangle$, $i\neq j$, in the quotient of ${\mathrm{PG}}(r-1,q)$ by $\chi(v_j)$.
The case $|J|>1$ follows by repeating the argument $|J|$ times.
\end{proof}

\begin{theorem}
Let $C$ be an additive MDS code.
The dual code $C^{\perp}$ is an additive faithful MDS code if and only if every non-obliterating geometric quotient of $C$ is faithful.
\end{theorem}
\begin{proof}
Let $C$ be an additive $[n,r/h,d]_{q}^h$ code, where $d=n-\lceil r/h \rceil+1$.
Suppose that every non-obliterating geometric quotient of $C$ is faithful and $C^{\perp}$ is not a faithful MDS code. If $C^{\perp}$ is unfaithful then by Theorem~\ref{degthm}, $C$ has minimum distance $1$, which contradicts the fact that $C$ is MDS.
Thus, it follows that $C^{\perp}$ is not MDS. Let $\tilde G$ be the $r \times hn$ matrix with entries in $\mathbb F_q$ obtained from a generator matrix $G$ for $C$ as before.
Since $C^\perp$ is not MDS there is a codeword $u  \in C^{\perp}$ of weight at most $m \leqslant \lfloor r/h \rfloor$. 
Consider the restriction of $\tilde G$ to these $m$ blocks of $h$ columns, and let $\pi_1,\ldots,\pi_m$ denote the corresponding elements of $\mathcal X(C)$ (subspaces of PG$(r-1,q)$). 
Since $u \in C^{\perp}$, at least one of the $\pi_i$ meets the span 
$$\Sigma= \langle \pi_j:j\in J\rangle, ~~ J=\{1,\ldots,m\}\setminus \{i\}$$ 
of the remaining elements nontrivially. Since $\Sigma$ has dimension at most $(m-1)h-1$, 
and $r-(m-1)h\geq h$, the quotient $C/J$ is non-obliterating.
Moreover the quotient $\pi_i/\Sigma$ has dimension at most $h-2$.
Thus, we have a non-obliterating quotient $C/J$ which is unfaithful.

Conversely, suppose $C^{\perp}$ is an additive faithful MDS code and that there is a non-obliterating geometric quotient $C/J$ of $C$, which is unfaithful. Suppose $J$ has size $m-1$.

Since the  geometric quotient  is non-obliterating $m-1 \leqslant \lfloor r/h \rfloor-1$. Since the  geometric quotient  is unfaithful, there is a further coordinate which together with the coordinates corresponding to $J$ give $m$ blocks whose $mh$ columns do not span an $mh$-dimensional subspace of ${\mathbb{F}}_q^r$. This implies that there is a codeword of weight $m\leqslant \lfloor r/h \rfloor$ in $C^{\perp}$, contradicting the assumption that $C^{\perp}$ is MDS.
\end{proof}

\section{Constructions of additive MDS codes over finite fields}

In the next theorem we will construct some additive MDS codes which prove that the bound in Theorem~\ref{dbound} is attainable if $r_0$ divides $h$ and $k=2$. Observe that Theorem~{\ref{dboundlinear} implies that for linear codes with $k=2$,
$$
n \leqslant q^h+1,
$$
so the following construction exceeds the bound for linear codes.

\begin{theorem} \label{k2cont} 
If $r_0$ divides $h$ then there  is a $[n,1+(r_0/h),n-1]_{q}^h$ additive MDS code where
$$
n=q^{h}+\frac{q^h-1}{q^{r_0}-1}.
$$
\end{theorem}

\begin{proof}
Let $A$ be the set of representatives for ${\mathbb F}_{q^{h+r_0}}/{\mathbb F}_{q^{r_0}}$, i.e. each $a \in A$ corresponds to a point of PG$(h/r_0,q^{r_0})$. 

Let $w$ be a primitive element of ${\mathbb F}_{q^{h+r_0}}$ and let $\{1,\alpha,\ldots,\alpha^{h/r_0-1} \}$ be a basis for ${\mathbb F}_{q^h}$ over ${\mathbb F}_{q^{r_0}}$.

For each $x \in {\mathbb F}_{q^{h+r_0}}$, we define a codeword whose coordinates are indexed by $a \in A$ and whose entry is
$$
\sum_{j=0}^{h/r_0-1}\mathrm{tr} (xaw^j) \alpha^j
$$
where $\mathrm{tr}$ is the trace from ${\mathbb F}_{q^{h+r_0}}$ to ${\mathbb F}_{q^{r_0}}$.

Thus $|C|=q^{h+r_0}$ and $n=(q^{h+r_0}-1)/(q^{r_0}-1)$, so it suffices to prove that $d=n-1$.

Equivalently, we need to prove that there is a unique $a$ such that 
$$
0=\sum_{j=0}^{h/r_0-1}\mathrm{tr} (xaw^j) \alpha^j
$$
Since $\{1,\alpha,\ldots,\alpha^{h/r_0-1} \}$ are linearly independent over ${\mathbb F}_{q^{r_0}}$ it follows that for all $j \in \{0,\ldots, h/r_0-1\}$,
$$
0=\mathrm{tr} (xaw^j).
$$
These equations define $h/r_0$ independent hyperplanes in PG$(h/r_0,q^{r_0})$, so there is a unique point which is in the intersection of these hyperplanes.
\end{proof}

A similar construction can be used to prove that the bound in Theorem~\ref{dbound} is also attainable when $k=3$ and $q=2$.

\begin{theorem} \label{k3cont} 
There  is a $[2^{h+1},2+(1/h),d]_{2}^h$ additive MDS code.
\end{theorem}

\begin{proof}

Let $w$ be a primitive element of ${\mathbb F}_{2^{h+1}}$ and let $\{1,\alpha,\ldots,\alpha^{h-1} \}$ be a basis for ${\mathbb F}_{2^h}$ over ${\mathbb F}_{2}$.

For each $(x_1,x_2) \in  {\mathbb F}_{2^{h}} \times {\mathbb F}_{2^{h+1}}$, we define a codeword whose coordinates are indexed by $a \in {\mathbb F}_{2^{h+1}}$ and whose entry is
$$
x_1+\sum_{j=0}^{h-1}\mathrm{tr} (x_2aw^j) \alpha^j
$$
where $\mathrm{tr}$ is the trace from ${\mathbb F}_{2^{h+1}}$ to ${\mathbb F}_{2}$.

Thus $|C|=2^{2h+1}$ and $n=2^{h+1}$, so it suffices to prove that $d=n-2$.

Suppose that there are $a_1,a_2 ,a_3$, distinct elements of ${\mathbb F}_{2^{h+1}}$ such that 
$$
0=x_1+\sum_{j=0}^{h-1}\mathrm{tr} (x_2a_iw^j) \alpha^j,
$$
for $ i \in \{1,2,3\}$. 

Observe that if $x_2=0$ then $x_1=0$ and we have the zero codeword. Thus, we assume $x_2 \neq 0$.

Thus,
$$
0=\sum_{j=0}^{h-1}\mathrm{tr} (x_2(a_1-a_2)w^j) \alpha^j=\sum_{j=0}^{h-1}\mathrm{tr} (x_2(a_1-a_3)w^j) \alpha^j,
$$

Since $\{1,\alpha,\ldots,\alpha^{h-1} \}$ are linearly independent over ${\mathbb F}_{2}$ it follows that 
$$
0=\mathrm{tr} (x_2(a_1-a_2)w^j)=\mathrm{tr} (x_2(a_1-a_3)w^j),
$$
for $j \in \{0,\ldots,h-1\}$.

The equations 
$$
0=\mathrm{tr} (x_2aw^j),
$$
for $j \in \{0,\ldots,h-1\}$, define $h$ independent hyperplanes in PG$(h,2)$, so there is a unique point $a$ which is in the intersection of these hyperplanes.

This implies
$$
x_2(a_1-a_2)=x_2(a_1-a_3).
$$
Since $x_2 \neq 0$, this implies $a_2=a_3$, a contradiction.
\end{proof}

Although the codes constructed in Theorem~\ref{k2cont} and Theorem~\ref{k3cont} have very small rate, it is of interest that their length is superior to that of their linear counterparts. In the next section we classify additive MDS codes over small fields and observe, once more, that there are additive fractional MDS codes whose length exceeds the length of any known linear MDS code.

In the previous two constructions, we have that $k$ is small. In the following construction, we look at the other extreme, when $d$ is small. 

\begin{theorem}  \label{partition}
Let $\pi_0$ and $\pi_{\infty}$ be $h$-dimensional subspaces of ${\mathbb F}_{q}^{2h}=\pi_0 \oplus \pi_{\infty}$. If there are $r_i$-dimensional subspaces $\pi_i$ such that $r_i \leqslant h$,
$$
\pi_i \cap \pi_j=\{0\}
$$
for all distinct $i,j \in \{0,1,\ldots,k \} \cup \{\infty\}$ where
$$
\sum_{i=1}^k {r_i}=r
$$
then there is an  $[ \lceil r/h \rceil +2,r/h,3]_{q}^h$ additive MDS code $C$.
\end{theorem}

\begin{proof}
We consider ${\mathbb F}_{q}^{2h}$ as ${\mathbb F}_{q^h}^{2}$ and choose a suitable basis so that
$$
\pi_0=\{(0,x) \ | \ x \in {\mathbb F}_{q^h}\} 
$$
and
$$
\pi_{\infty}= \{(x,0) \ | \ x \in {\mathbb F}_{q^h}\}.
$$
 
Let  $S_j$ be an additive subset of ${\mathbb F}_{q^h}$ of size $q^{r_j}$, for each $j \in \{1,\ldots,k\}$.

We define additive maps $f$ and $g$ from
$$
(S_1,\ldots,S_{k}) \rightarrow {\mathbb F}_{q^h}
$$
such that
$$
\pi_j=\{(f(0,\ldots,0,x_j,0,\ldots,0),g(0,\ldots,0,x_j,0,\ldots,0)) \ | \ x \in S_j\},
$$
We define an additive code
$$
C=\{ (u_1,u_2,\ldots,u_{k},f(u_1,\ldots,u_{k}),g(u_1,\ldots,u_{k})) \ | \ u_i \in S_i \}.
$$
and note that $|C|=q^r$.

Let $u \in C$ be a codeword. We need to show that the weight of $u$ is at least $3$.

If the weight of $(u_1,\ldots,u_k)$ is at least $3$ then the weight of $u$ is at least $3$.

If the weight of $(u_1,\ldots,u_k)$ and the weight $u$ is $2$ then we can suppose that there exists distinct $i$ and $j$ such that $u_i, u_j \neq 0$ and 
$$
f(0,\ldots,0,u_i,0,\ldots,0,u_j,0,\ldots,0)=g(0,\ldots,0,u_i,0,\ldots,0,u_j,0,\ldots,0)=0.
$$
Since $f$ and $g$ are additive, we have that
$$
f(0,\ldots,0,u_i,0,\ldots,0)=f(0,\ldots,0,-u_j,0,\ldots,0)
$$
and likewise for $g$, which implies that $\pi_i$ and $\pi_j$ have non-trivial intersection.
Since this is not the case the weight of $u$ is at least $3$.

If the weight of $(u_1,\ldots,u_k)$ is $1$ then suppose that $u_i \neq 0$ and that
$$
f(0,\ldots,0,u_i,0,\ldots,0)=0.
$$
This implies that $\pi_i$ and $\pi_0$ have non-trivial intersection.
Since this is not the case the weight of $u$ is at least $3$. The same argument holds for $g$ and $\pi_{\infty}$.

Therefore, $C$ has minimum weight $3$, which implies that it has minimum distance $3$.
\end{proof}


\begin{theorem} 
There is an  $[\lceil r/h \rceil +2,r/h,3]_{q}^h$ additive MDS code $C$ for all $\lceil r/h \rceil \leqslant q^h-1$.
\end{theorem}

\begin{proof}
Let $r=(k-1)h+r_0$.

The non-zero vectors of ${\mathbb F}_q^{2h}$ can be partitioned in to $q^h+1$ subspaces $\pi_0,\ldots,\pi_{q^h}$ of dimension $h$. We can then fix $\pi_0$ and $\pi_{\infty}$ to be two of these subspaces and choose $r_i=h-1$ for $i=1,\ldots,h-r_0$ and $r_i=h$ for $i=k-h+r_0$ using the remaining $q^h-1$ subspaces.

Then we have that
$$
\sum_{i=1}^k r_i=h(k-1)+r_0
$$
and the theorem follows from Theorem~\ref{partition}.
\end{proof}

\section{Additive MDS codes over small finite fields}

In this section we aim to classify all additive MDS codes over fields of size $q \leqslant 9$ and small enough values for $r$.
In the integral case, these codes had already been classified in \cite{BGL2023}.
We continue this classification here, considering the fractional case too. We will only consider faithful MDS codes, since an unfaithful MDS code can always be converted in a faithful MDS code with the same parameters, as mentioned in Remark~\ref{onlyrem}. We do this by classifying the projective systems corresponding to (fractional) additive codes as detailed in Theorem \ref{thm:proj_system}. In particular, a faithful $[n,r/h,d]_q^h$ MDS code is equivalent to a set of $(h-1)$-dimensional subspaces of size $n$, any $\lceil r/h \rceil$ of which span PG$(r-1,q)$.
We may assume $n \geqslant \lceil r/h \rceil +1$, since otherwise the condition is trivially satisfied.

A projective system corresponding to a fractional additive $[n,r/h,d]_q^h$ code is called a {\em fractional subspace-arc}. A fractional subspace-arc is {\em complete} if it can not be extended to a larger fractional subspace-arc. If $h=2$ we speak of {\em line-arcs} and if $h=3$ of {\em plane-arcs}.

In the following tables, the integral MDS codes are indicated by $r$ in bold face. Constructions obtained from Theorem~\ref{k2cont} are indicated by a $\dagger$, and constructions obtained from Theorem~\ref{k3cont} are indicated by a $*$.

The computational results were obtained using the GAP-package FinInG \cite{GAP,FinInG}.

\subsection{Additive MDS codes over ${\mathbb F}_4$}

In this case, $q=2$ and $h=2$. By Theorem~\ref{kbound}, we have $k \leqslant 3$, so $r \leqslant 6$, and 
by Theorem~\ref{dbound}, $n \leqslant 8$ if $r=5$ and $n \leqslant 6$ if $r=6$. The number of equivalence classes of (fractional) line-arcs are collected in Table \ref{codes4}.
\begin{table}[h] 
\begin{center}
\begin{tabular}{|r|c|c|c|c|c|c|} \hline
length $n$ & 3 & 4 & 5 & 6 & 7 & 8 \\ \hline
$r=3$ & 2&  2&  1&   1&  $1^{\dagger}$ & -  \\ \hline
$\mathbf{r=4}$ &1 & 1 & 1 & - &- & - \\ \hline
$r=5$ & - & 5&  8&  6& 1& $1^{*}$ \\ \hline
$\mathbf{r=6}$ & - & 1 &1 & 1 & - & - \\ \hline
\end{tabular}
\vspace{0.3cm}
\caption{The number of additive MDS codes over ${\mathbb F}_4$.}\label{codes4}
\end{center}
\end{table}

The most interesting case is that of fractional line-arcs in ${\mathrm{PG}}(4,2)$, where the maximal size is 8. Amongst the representatives for the equivalent classes of these line-arcs, the only complete fractional line-arcs of size strictly smaller than 8 have size 6 and there are 5 of them. So there is only one non-complete fractional line-arc of size 6. This implies that each two of the fractional line-arcs of size 6 contained in the unique fractional line-arc of size 8 are equivalent.

\subsection{Additive MDS codes over ${\mathbb F}_8$}

In this case $q=2$ and $h=3$. So the projective systems are plane-arcs in ${\mathrm{PG}}(r-1,2)$.
By Theorem~\ref{kbound}, we have $k \leqslant 7$, if $r_0 \in \{2,3\}$ and $k \leqslant 8$ if $r_0=1$. 
Table~\ref{codes8} contains the number of additive MDS codes for $r \leqslant 9$.

\begin{table}[h] 
\begin{center}
\begin{tabular}{|r|c|c|c|c|c|c|c|c|c|c|c|c|c|c|} \hline
length $n$ & 3 & 4 & 5 & 6 & 7 & 8 & 9 & 10 & 11 & 12 & 13 & 14 & 15 & 16 \\ \hline
$r=4$ & 2&  3&  4&  5& 6& 6& 5& 4& 3& 2& 1& 1& 1 ($\geqslant 1^{\dagger})$ &  - \\ \hline
$r=5$ & 2&  4&  10&  14& 19& 9& 4&0 & -&- &- &- &- &  - \\ \hline
$\mathbf{r=6}$ & 1&  1&  2&  1& 1& 1& 1& -&- & -& -&- &- & -  \\ \hline
$r=7$ &- &  { 33}& { 1895} &  & & & & & & & & & &  $\geqslant 1^*$  \\ \hline
$r=8$ & -&  { 6}&  { 165}&  { 15480}& & & & & & -&- &- &- & -  \\ \hline
$\mathbf{r=9}$ &- & 1 & 2 & 4 & 2 & 2 & 2 & 1  & - & - &  - & - & - & -\\ \hline
\end{tabular}
\vspace{0.3cm}
\caption{The number of MDS additive codes over ${\mathbb F}_8$.}\label{codes8}
\end{center}
\end{table}

The maximal size of a plane-arc in ${\mathrm{PG}}(5,2)$ ($r=6$) is 9 (integral case). Up to equivalence, there is a unique example for each size except in the case that $n= 5$, for which there are two inequivalent examples.
As of now, the classification of fractional plane-arcs in ${\mathrm{PG}}(6,2)$ and ${\mathrm{PG}}(7,2)$ could not yet be completed.

\subsection{Additive MDS codes over ${\mathbb F}_9$}

The projective systems corresponding to these ternary codes are (fractional) line-arcs in ${\mathrm{PG}}(r-1,3)$.
By Theorem~\ref{kbound}, we have $k \leqslant 8$, so $r \leqslant 16$. Table~\ref{codes9} details the number of additive MDS codes for $r \leqslant 6$.

\begin{table}[h] 
\begin{center}
\begin{tabular}{|r|c|c|c|c|c|c|c|c|c|c|c|c|c|c|} \hline
length $n$ & 3 & 4 & 5 & 6 & 7 & 8 & 9 & 10 & 11 & 12 & 13 & 14 & 15 & 16 \\ \hline
$r=3$ & 2&  3&  3&  4& 4& 3& 3& 2&1& 1& $1^{\dagger}$ & -& -&  - \\ \hline
$\mathbf{r=4}$ & 1&  3&  4&  5& 4& 3& 2& 2& -&- &- &- &- &  - \\ \hline
$r=5$ & {  2}&  {  6}&  {  48}&  {  1167}& {  21248}& {  145451}& {  273753}& {  96854}&{   3039}&6 &0 & 0& -&  - \\ \hline
$\mathbf{r=6}$ & -& 1 & 4 & 13 & 4 & 3 & 1 & 1 & 0 & -& -  &- &- &- \\ \hline
\end{tabular}
\vspace{0.3cm}
\caption{The number of MDS additive codes over ${\mathbb F}_9$.}\label{codes9}
\end{center}
\end{table}
The maximal size of a fractional line-arc in ${\mathrm{PG}}(4,3)$ is 12, and there are 6 examples up to projective equivalence.
The number of examples of sizes 1 up to 12 is as in the following list
$$
[1, 2, 2, 6, 48, 1167, 21248, 145451, 273753, 96854, 3039, 6 ]
$$
up to projective equivalence.
Amongst these, the numbers of complete fractional line-arcs in ${\mathrm{PG}}(4,3)$ of sizes 1 up to 12 is as follows
$$
[0, 0, 0, 0, 0,  0, 0, 342, 21787, 68725, 3003, 6 ].
$$
So the smallest complete fractional line-arcs have size 8.
\begin{remark} \label{12arc}
Interestingly, each of the fractional line-arcs in ${\mathrm{PG}}(4,3)$ of maximum size 12 consists of pairwise disjoint lines. However, there are complete fractional line-arcs of smaller size which do not satisfy this property. The number of ${\mathrm{PGL}}$-equivalence classes of fractional line-arcs  in ${\mathrm{PG}}(4,3)$ which consist of pairwise disjoint lines is
$$
[1, 1, 1, 3, 27, 607, 12386, 100185, 227659, 91720, 3013, 6],
$$
and amongst these, the number of complete fractional line-arcs is
$$
[0, 0, 0, 0, 0, 0, 0, 0, 2802, 63788, 2977, 6].
$$
So the smallest complete fractional line-arc in ${\mathrm{PG}}(4,3)$ consisting of pairwise disjoint lines has size 9.
\end{remark}

{\bf Representatives of fractional line-arcs.}
Here we list the representatives of the six equivalence classes of fractional line-arcs of size 12 in ${\mathrm{PG}}(4,3)$ up to projective equivalence. Each list consists of twelve $(2\times 5)$ matrices with entries from ${\mathbb{F}}_3$, whose rows span a line in ${\mathrm{PG}}(4,3)$.
\[
A_1:\begin{matrix}
    \left[\begin{smallmatrix} 10000 \\ 01000 \end{smallmatrix}\right],
    \left[\begin{smallmatrix} 10001 \\ 01010 \end{smallmatrix}\right],
    \left[\begin{smallmatrix} 10002 \\ 01100 \end{smallmatrix}\right],
    \left[\begin{smallmatrix} 10010 \\ 01101 \end{smallmatrix}\right],
    \left[\begin{smallmatrix} 10011 \\ 01200 \end{smallmatrix}\right],
    \left[\begin{smallmatrix} 10012 \\ 01210 \end{smallmatrix}\right],
    \left[\begin{smallmatrix} 10100 \\ 01012 \end{smallmatrix}\right],
    \left[\begin{smallmatrix} 10112 \\ 01222 \end{smallmatrix}\right],
    \left[\begin{smallmatrix} 10122 \\ 01002 \end{smallmatrix}\right],
    \left[\begin{smallmatrix} 10200 \\ 01220 \end{smallmatrix}\right],
    \left[\begin{smallmatrix} 10220 \\ 01211 \end{smallmatrix}\right],
    \left[\begin{smallmatrix} 01011 \\ 00112 \end{smallmatrix}\right]
\end{matrix}
\]
\[
A_2: \begin{matrix}
    \left[\begin{smallmatrix} 10000 \\ 01000 \end{smallmatrix}\right],
    \left[\begin{smallmatrix} 10001 \\ 01010 \end{smallmatrix}\right],
    \left[\begin{smallmatrix} 10002 \\ 01100 \end{smallmatrix}\right],
    \left[\begin{smallmatrix} 10010 \\ 01101 \end{smallmatrix}\right],
    \left[\begin{smallmatrix} 10011 \\ 01200 \end{smallmatrix}\right],
    \left[\begin{smallmatrix} 10012 \\ 01210 \end{smallmatrix}\right],
    \left[\begin{smallmatrix} 10100 \\ 01012 \end{smallmatrix}\right],
    \left[\begin{smallmatrix} 10121 \\ 01001 \end{smallmatrix}\right],
    \left[\begin{smallmatrix} 10220 \\ 01122 \end{smallmatrix}\right],
    \left[\begin{smallmatrix} 10222 \\ 01112 \end{smallmatrix}\right],
    \left[\begin{smallmatrix} 11022 \\ 00102 \end{smallmatrix}\right],
    \left[\begin{smallmatrix} 12010 \\ 00112 \end{smallmatrix}\right]
\end{matrix}
\]
\[
A_3:\begin{matrix}
    \left[\begin{smallmatrix} 10000 \\ 01000 \end{smallmatrix}\right],
    \left[\begin{smallmatrix} 10001 \\ 01010 \end{smallmatrix}\right],
    \left[\begin{smallmatrix} 10002 \\ 01100 \end{smallmatrix}\right],
    \left[\begin{smallmatrix} 10010 \\ 01101 \end{smallmatrix}\right],
    \left[\begin{smallmatrix} 10011 \\ 01200 \end{smallmatrix}\right],
    \left[\begin{smallmatrix} 10012 \\ 01210 \end{smallmatrix}\right],
    \left[\begin{smallmatrix} 10110 \\ 01012 \end{smallmatrix}\right],
    \left[\begin{smallmatrix} 10201 \\ 01112 \end{smallmatrix}\right],
    \left[\begin{smallmatrix} 10221 \\ 01111 \end{smallmatrix}\right],
    \left[\begin{smallmatrix} 10021 \\ 00121 \end{smallmatrix}\right],
    \left[\begin{smallmatrix} 12010 \\ 00110 \end{smallmatrix}\right],
    \left[\begin{smallmatrix} 01002 \\ 00120 \end{smallmatrix}\right]
\end{matrix}
\]
\[
A_4:\begin{matrix}
    \left[\begin{smallmatrix} 10000 \\ 01000 \end{smallmatrix}\right],
    \left[\begin{smallmatrix} 10001 \\ 01010 \end{smallmatrix}\right],
    \left[\begin{smallmatrix} 10002 \\ 01100 \end{smallmatrix}\right],
    \left[\begin{smallmatrix} 10010 \\ 01101 \end{smallmatrix}\right],
    \left[\begin{smallmatrix} 10011 \\ 01200 \end{smallmatrix}\right],
    \left[\begin{smallmatrix} 10012 \\ 01211 \end{smallmatrix}\right],
    \left[\begin{smallmatrix} 10110 \\ 01011 \end{smallmatrix}\right],
    \left[\begin{smallmatrix} 10121 \\ 01112 \end{smallmatrix}\right],
    \left[\begin{smallmatrix} 10021 \\ 00102 \end{smallmatrix}\right],
    \left[\begin{smallmatrix} 12010 \\ 00120 \end{smallmatrix}\right],
    \left[\begin{smallmatrix} 12022 \\ 00101 \end{smallmatrix}\right],
    \left[\begin{smallmatrix} 10102 \\ 00010 \end{smallmatrix}\right]

\end{matrix}
\]
\[
A_5:\begin{matrix}
    \left[\begin{smallmatrix} 10000 \\ 01000 \end{smallmatrix}\right],
    \left[\begin{smallmatrix} 10001 \\ 01010 \end{smallmatrix}\right],
    \left[\begin{smallmatrix} 10002 \\ 01100 \end{smallmatrix}\right],
    \left[\begin{smallmatrix} 10010 \\ 01101 \end{smallmatrix}\right],
    \left[\begin{smallmatrix} 10011 \\ 01200 \end{smallmatrix}\right],
    \left[\begin{smallmatrix} 10012 \\ 01221 \end{smallmatrix}\right],
    \left[\begin{smallmatrix} 10100 \\ 01020 \end{smallmatrix}\right],
    \left[\begin{smallmatrix} 10101 \\ 01211 \end{smallmatrix}\right],
    \left[\begin{smallmatrix} 10102 \\ 01122 \end{smallmatrix}\right],
    \left[\begin{smallmatrix} 10120 \\ 01111 \end{smallmatrix}\right],
    \left[\begin{smallmatrix} 10122 \\ 01212 \end{smallmatrix}\right],
    \left[\begin{smallmatrix} 01011 \\ 00112 \end{smallmatrix}\right]
\end{matrix}
\]
\[
A_6:\begin{matrix}
    \left[\begin{smallmatrix} 10000 \\ 01000 \end{smallmatrix}\right],
    \left[\begin{smallmatrix} 10001 \\ 01010 \end{smallmatrix}\right],
    \left[\begin{smallmatrix} 10002 \\ 01100 \end{smallmatrix}\right],
    \left[\begin{smallmatrix} 10010 \\ 01101 \end{smallmatrix}\right],
    \left[\begin{smallmatrix} 10011 \\ 01200 \end{smallmatrix}\right],
    \left[\begin{smallmatrix} 10012 \\ 01222 \end{smallmatrix}\right],
    \left[\begin{smallmatrix} 10200 \\ 01020 \end{smallmatrix}\right],
    \left[\begin{smallmatrix} 10201 \\ 01221 \end{smallmatrix}\right],
    \left[\begin{smallmatrix} 10202 \\ 01011 \end{smallmatrix}\right],
    \left[\begin{smallmatrix} 10220 \\ 01211 \end{smallmatrix}\right],
    \left[\begin{smallmatrix} 10221 \\ 01121 \end{smallmatrix}\right],
    \left[\begin{smallmatrix} 10222 \\ 01112 \end{smallmatrix}\right]
\end{matrix}
\]

{\bf Fractional additive MDS codes.}
Below we list the generator matrix $G_i\in {\mathbb{F}}_9^{5\times 12}$ of the fractional additive $[12,5/2,10]_3^2$ MDS code $C_i$ corresponding to the fractional line-arc $A_i$, $i=1,\ldots,6$, listed above. The element $\omega$ is a primitive element of ${\mathbb{F}}_9$ with minimal polynomial $x^2-x-1$.
\[G_1 = \left(
\begin{smallmatrix}
1 & 1 & 1 & 1 & 1 & 1 & 1 & 1 & 1 & 1 & 1 & 0 \\
\omega & \omega & \omega & \omega & \omega & \omega & \omega & \omega & \omega & \omega & \omega & 1 \\
0 & 0 & \omega & \omega & \omega^5 & \omega^5 & 1 & \omega^3 & 1 & \omega^6 & \omega^6 & \omega \\
0 & \omega & 0 & 1 & 1 & \omega^2 & \omega & \omega^3 & \omega & \omega^5 & \omega^7 & \omega^2 \\
0 & 1 & \omega & \omega & 1 & \omega & \omega^5 & \omega^6 & \omega^6 & 0 & \omega & \omega^3 \\
\end{smallmatrix}
\right)
\]

\[G_2 = \left(
\begin{smallmatrix}
1 & 1 & 1 & 1 & 1 & 1 & 1 & 1 & 1 & 1 & 1 & 1 \\
\omega & \omega & \omega & \omega & \omega & \omega & \omega & \omega & \omega & \omega & 1 & \omega \\
0 & 0 & \omega & \omega & \omega^5 & \omega^5 & 1 & 1 & \omega^7 & \omega^7 & \omega & \omega \\
0 & \omega & 0 & 1 & 1 & \omega^2 & \omega & 1 & \omega^6 & \omega^7 & 1 & \omega^2 \\
0 & 1 & \omega & \omega^2 & 1 & \omega & \omega^5 & \omega^2 & \omega^5 & \omega^6 & \omega^6 & \omega^5 \\
\end{smallmatrix}
\right)
\]

\[G_3 = \left(
\begin{smallmatrix}
1 & 1 & 1 & 1 & 1 & 1 & 1 & 1 & 1 & 1 & 1 & 0 \\
\omega & \omega & \omega & \omega & \omega & \omega & \omega & \omega & \omega & 0 & 1 & \omega \\
0 & 0 & \omega & \omega & \omega^5 & \omega^5 & 1 & \omega^7 & \omega & \omega & \omega & 1 \\
0 & \omega & 0 & 1 & 1 & \omega^2 & \omega^2 & \omega & \omega^7 & \omega^6 & \omega^2 & \omega^5 \\
0 & 1 & \omega & \omega^2 & 1 & \omega & \omega^5 & \omega^3 & \omega^2 & \omega^2 & 0 & \omega \\
\end{smallmatrix}
\right)
\]

\[G_4 = \left(
\begin{smallmatrix}
1 & 1 & 1 & 1 & 1 & 1 & 1 & 1 & 1 & 1 & 1 & 0 \\
\omega & \omega & \omega & \omega & \omega & \omega & \omega & \omega & \omega & \omega & \omega & 1 \\
0 & 0 & \omega & \omega & \omega^5 & \omega^5 & \omega & \omega^6 & \omega & \omega^6 & \omega^7 & \omega^7 \\
0 & \omega & 0 & 1 & 1 & \omega^3 & \omega^5 & \omega^5 & \omega & \omega^7 & \omega^6 & \omega^7 \\
0 & 1 & \omega & \omega^2 & 1 & \omega^6 & 0 & \omega^2 & \omega^7 & \omega & \omega^2 & \omega^6 \\
\end{smallmatrix}
\right)
\]

\[G_5 = \left(
\begin{smallmatrix}
1 & 1 & 1 & 1 & 1 & 1 & 1 & 1 & 1 & 1 & 1 & 0 \\
\omega & \omega & \omega & \omega & \omega & \omega & \omega & \omega & \omega & \omega & \omega & 1 \\
0 & 0 & \omega & \omega & \omega^5 & \omega^5 & 1 & \omega^3 & \omega^2 & \omega^2 & \omega^3 & \omega \\
0 & \omega & 0 & 1 & 1 & \omega^3 & \omega^5 & \omega^5 & \omega & \omega^7 & \omega^6 & \omega^7 \\
0 & 1 & \omega & \omega^2 & 1 & \omega^6 & 0 & \omega^2 & \omega^7 & \omega & \omega^2 & \omega^6 \\
\end{smallmatrix}
\right)
\]
\[
G_6 = \left(
\begin{smallmatrix}
1 & 1 & 1 & 1 & 1 & 1 & 1 & 1 & 1 & 1 & 1 & 0 \\
\omega & \omega & \omega & \omega & \omega & \omega & \omega & \omega & \omega & \omega & \omega & 1 \\
0 & 0 & \omega & \omega & \omega^5 & \omega^5 & \omega & \omega^6 & \omega & \omega^6 & \omega^7 & \omega^7 \\
0 & \omega & 0 & 1 & 1 & \omega^3 & \omega^5 & \omega^5 & \omega & \omega^7 & \omega^6 & \omega^7 \\
0 & 1 & \omega & \omega & 1 & \omega & \omega^5 & \omega^3 & \omega^2 & \omega^2 & 0 & \omega \\
\end{smallmatrix}
\right)
\]

\section{Conclusion}
In this article we introduced new bounds relating the parameters of an additive code over a finite field, Theorem~\ref{addbound} and Theorem~\ref{addbound2}. These bounds have similarities to the Griesmer bound for linear codes and can be obtained in a combinatorial manner. It is possible that similar bounds exist which may improve on our bounds here. In the case of MDS codes however, it is unlikely that any improvement would be obtained using solely combinatorial arguments. 
We have provided examples of fractional additive MDS codes which attain the bound, Theorem~\ref{k2cont} and Theorem~\ref{k3cont}. We have also found examples over small fields which surpass the length of Reed-Solomon codes, most notably the $[12,2.5,10]_3^2$ additive codes  mentioned in Remark~\ref{12arc}. It would be of great interest if there are families of fractional MDS codes which are longer than the Reed-Solomon codes for larger values of $r/h$.

The length of the longest $k$-dimensional linear MDS codes is conjectured to be $q+1$ apart from the exceptional cases that $k \in \{3,q-1\}$, $q$ is even and $n=q+2$. The MDS conjecture was proved for $q$ prime in \cite{Ball2012}. For a recent list on the values of $k$ for which the conjecture is known to be true for $q$ non-prime, see \cite{BL2019}. 
It is probable that the MDS conjecture is true for all MDS codes over any alphabet of size $q$, i.e.  that $n\leqslant q+1$, where the code has size $q^{k}$, apart from the previously mentioned exceptional cases that  $k \in \{3,q-1\}$, $q$ is even and $n=q+2$. Note that, for $q$ odd, three-dimensional linear MDS codes of length $q$ and $q + 1$ were classified by Segre in 1955, while the classification of linear MDS codes of length $q-1$ and $q-2$ was only recently completed, see \cite{BL2018}.

There are examples of additive integral non-MDS codes which outperform their linear counterparts constructed using cyclic codes in \cite{GLLM2023}. For example a $[63,5,45]_2^2$ additive code over ${\mathbb F}_4$ is constructed and no linear code is known with these parameters. This additive code is not equivalent to a linear code. Note that a linear code with such parameters might exist. To our knowledge, the only integral additive codes with a parameter set for which we know no such linear code exists are Mathon's $[21,3,18]_3^2$ codes over ${\mathbb F}_9$ \cite{DDHM2002}, and several codes over ${\mathbb F}_4$ constructed from large sets of skew lines in projective space due to Kurz \cite{Kurz2024}. For example, there is a $[89,4,66]_2^2$ code over ${\mathbb F}_4$, but no $[89,4,66]_4^1$ (linear) code exists according to the tables in \cite{Grassl}.

\section{Acknowledgements}

We are grateful to the referees and to Zhiyu Yuan for their comments and suggestions. The referee's comments
led to considerable improvements to this article. We are grateful to Yue Zhou for notifying us that 
Theorem 16  was incorrect as stated in earlier versions of this article. We would like to thank Francesco Pavese for pointing out that their were errors in the data of the
fractional line-arcs of size 12 in ${\mathrm{PG}}(4,3)$. These errors have been rectified in this version.

\end{document}